\documentclass[preprint,3p,times,twocolumn,sort&compress]{elsarticle}

\usepackage{amssymb}
\usepackage{amsmath}

\usepackage{amsmath,amssymb,amsfonts}
\usepackage{algorithmic,algorithm}
\usepackage{graphicx}
\usepackage{textcomp}

\usepackage{amsthm,amssymb}
\usepackage{mathrsfs}
\usepackage{enumerate}
\usepackage{float}
\usepackage{color}
\usepackage{bm}
\usepackage{subcaption}
\usepackage{caption}
\usepackage{tikz}
\usetikzlibrary{arrows.meta, decorations.pathmorphing}
\usetikzlibrary{decorations.pathmorphing, patterns, arrows.meta}
\usetikzlibrary{decorations.pathmorphing, patterns, calc}
\usepackage{amsmath}

\usepackage{hyperref}
\usepackage{algorithm}
\usepackage{footnote}
\usepackage{multirow}
\usepackage{setspace}
\usepackage{balance} 

\makeatletter

\newcommand{\Rmnum}[1]{\expandafter\@slowromancap\romannumeral #1@}

\newtheorem{assumption}{Assumption}
\newtheorem{lemma}{Lemma}
\newtheorem{remark}{Remark}
\newtheorem{theorem}{Theorem}
\newtheorem{definition}{Definition}

\newtheorem{corollary}{Corollary}

\journal{Asian Journal of Control}

\begin{document}

\begin{frontmatter}



\title{Finite Sample Analysis of Subspace Identification for Stochastic Systems}


\author[1]{Shuai Sun\corref{cor1}%
}
\ead{suns19@mails.tsinghua.edu.cn}
\author[2]{Weikang Hu}
\ead{wk.hu@connect.ust.hk}
\author[3]{Xu Wang}
\ead{worm@mail.ustc.edu.cn}
\cortext[cor1]{Corresponding author}

\affiliation[1]{organization={Department of Automation, Tsinghua University},
city={Beijing},
postcode={100084},
country={China}}
\affiliation[2]{organization={Division of Emerging Interdisciplinary Areas, Hong Kong University of Science and Technology},
city={Hong Kong},
country={China}}
\affiliation[3]{organization={School of Computer Science and Technology, University of Science and Technology of China},
city={Hefei,Anhui},
postcode={230026},
country={China}}


\begin{abstract}
The subspace identification method (SIM) has become a widely adopted approach for the identification of discrete-time linear time-invariant (LTI) systems. In this paper, we derive finite sample high-probability error bounds for the system matrices $A,C$, the Kalman filter gain $K$ and the estimation of system poles. Specifically, we demonstrate that, ignoring the logarithmic factors, for an $n$-dimensional LTI system with no external inputs, the estimation error of these matrices decreases at a rate of at least $ \mathcal{O}(\sqrt{1/N}) $, while the estimation error of the system poles decays at a rate of at least $ \mathcal{O}(N^{-1/2n}) $, where $ N $ represents the number of sample trajectories. Furthermore, we reveal that achieving a constant estimation error requires a super-polynomial sample size in $n/m $, where $n/m$ denotes the state-to-output dimension ratio. Finally, numerical experiments are conducted to validate the non-asymptotic results.
\end{abstract}




\end{frontmatter}



\section{Introduction}

Linear time-invariant (LTI) systems are an important class of models with broad applications in finance, biology, robotics, and other engineering fields~\cite{ljung1986system}. The identification of LTI state-space models from input/output sample data is one of the core problems in system analysis, design and control~\cite{michel1992}. Among the various approaches for system identification, subspace identification method (SIM)~\cite{deistler1995consistency} and the least-squares method~\cite{ljung1976consistency} are two widely used techniques. This paper focuses primarily on SIM due to the following advantages~\cite{michel1992}: \textbf{\romannumeral1)} simplified model selection, \textbf{\romannumeral2)} seamless extension from single-input single-output (SISO) to multiple-input multiple-output (MIMO) systems, \textbf{\romannumeral3)} numerical stability and computational efficiency, with no need for iterative optimization, making it ideal for large-scale data and online applications.


 The SIMs mainly involve extracting relevant subspaces from the input/output data matrix through projections or regressions and use them to recover the corresponding state-space model. A number of variants of SIMs have been published, such as canonical variate analyis (CVA)~\cite{larimore1990canonical}, numerical algorithms for state space system identification (N4SID)~\cite{van1994n4sid}, and multivariable output-error state space method (MOESP)~\cite{verhaegen1994identification}. Over the past three decades, SIMs have attracted significant research interest and achieved substantial development, and numerous meaningful progress~\cite{van1995unifying,ljung1996subspace,chiuso2005consistency,massioni2008subspace,haber2014subspace,sinquin2018k4sid,cox2021linear,kedia2022fast} has been made. Therefore, analyzing the identification error of SIMs has become of critical importance.

To date, numerous papers~\cite{deistler1995consistency,peternell1996statistical,viberg1997analysis,bauer1997analysis,jansson1998consistency,bauer1999consistency,bauer2000analysis,bauer2002some,bauer2005asymptotic,jansson2000asymptotic,knudsen2001consistency,chiuso2004asymptotic} have analyzed the estimation error of SIMs from an asymptotic perspective. These works primarily investigate the asymptotic properties of SIMs, including consistency~\cite{jansson1998consistency}, asymptotic normality~\cite{bauer1999consistency}, and variance analysis~\cite{jansson2000asymptotic}. However, in practical applications, only a finite amount of sample data is available. Therefore, deriving sharp error bounds under finite sample conditions is essential.


The LTI systems can generally be classified as either \emph{fully observed} or \emph{partially observed} systems, depending on whether direct state measurements are available. For \textit{fully observed} LTI systems, where the state can be accurately measured, some studies~\cite{faradonbeh2018finite,simchowitz2018learning,sarkar2019near,matni2019tutorial} have established finite-sample upper bounds on the identification error of the system matrices $(A,B)$. However, these studies primarily focus on the least-squares method framework rather than SIMs. The identification problem for \textit{partially observed} LTI systems is more challenging, as the state cannot be directly measured~\cite{zheng2021sample}. Recent studies~\cite{tsiamis2019finite, simchowitz2019learning, zheng2020non, oymak2021revisiting} have investigated non-asymptotic identification and provided convergence rates for the realization of the state-space model parameters $(A,B,C,D)$, up to a similarity transformation, with rates of $\mathcal{O}(1/\sqrt{T}) $ or $\mathcal{O}(1/\sqrt{N}) $, subject to logarithmic factors, \( T \) and \( N \) represent the length and the number of sample trajectories, respectively. These results rely on the assumption that the system is persistently excited by process noise or external inputs.



In this paper, we analyze the identification errors  for $n$-dimensional discrete-time LTI stochastic systems with $m$ outputs and no external input, using SIM with multiple trajectories of equal length. We provide end-to-end estimation guarantees for the system matrices $A$, $C$, the Kalman filter gain $K$, and the system poles. Similar to~\cite{tsiamis2019finite, faradonbeh2018finite, simchowitz2018learning, sarkar2019near, fattahi2019learning, oymak2021revisiting, simchowitz2019learning, weyer1999finite, campi2002finite, vidyasagar2008learning}, this paper focuses on data-independent bounds, specifically investigating how the identification error depends on the number of sample trajectories $N$, as well as the system and algorithm parameters. Furthermore, we examine the complexity of the sample size required to ensure, with high probability, that the estimation error remains within a constant order of magnitude.

The main distinctions between this paper and~\cite{tsiamis2019finite} are as follows: \textbf{\romannumeral1}) The tools used in this paper mainly involve random matrix theory and matrix analysis, without relying on self-normalized martingales as employed in~\cite{tsiamis2019finite}. This difference simplifies the analysis in the present work. 
\textbf{\romannumeral2}) In addition to the system matrices $A$, $C$, the Kalman filter gain $K$, we also analyze the finite sample identification error of the system poles, which is not addressed in~\cite{tsiamis2019finite}. 
\textbf{\romannumeral3}) Leveraging singular value theory, we analyze the complexity of the sample size required to ensure, with high probability, that the estimation error remains within a constant order of magnitude. This aspect is not considered in~\cite{tsiamis2019finite}.

Our main contributions are summarized as follows: 
\begin{itemize}
	\item We derive finite sample high-probability error bounds for the system matrices $A,C$, the Kalman filter gain $K$ and the estimation of system poles. Specifically, we demonstrate that, ignoring the logarithmic factors, the error of these matrices decreases at a rate of at least $ \mathcal{O}(\sqrt{1/N}) $ with $N$ independent sample trajectories, while the error of the system poles decays at a rate of at least $ \mathcal{O}(N^{-1/2n}) $.

	\item We reveal that achieving a constant estimation error requires a super-polynomial sample size in $n/m $, where $n/m$ denotes the state-to-output dimension ratio.

\end{itemize}

This paper is organized as follows.
Section~\ref{sec:problem setup} formulates the identification problem and reviews the SIM. Section~\ref{sec:system matrix} and Section~\ref{sec:system pole} analyze the finite sample identification error of the system matrices and the system poles, respectively.
Finally, Section~\ref{sec:numerical results} provides numerical simulations.

\textbf{Notations}: 
$\mathbf{0}$ is an all-zero matrix. For any $ x \in \mathbb{R} $, $ \lfloor x \rfloor $ denotes the largest integer not exceeding $ x $, and $ \lceil x \rceil $ denotes the smallest integer not less than $ x $.
The Frobenius norm is denoted by $ \|A\|_{\rm F} $, and $ \|A\| $ is the spectral norm of $A$, i.e., its largest singular value $ \sigma_{\max}(A) $. 
Its $ j $-th largest singular value is denoted by $ \sigma_{j}(A) $, and its smallest non-zero singular value  is denoted by $ \sigma_{\min}(A) $.
 The spectrum of a square matrix $ A $ is denoted by $ \lambda(A) $.
The Moore-Penrose inverse of matrix $ A $ is denoted by $ A^\dagger $. 
The $ n \times n $ identity matrix is denoted by $\mathbf{I}_n$, and the identity matrix of appropriate dimension is denoted by $\mathbf{I}$.
The Kronecker product is denoted by $\otimes$.
The matrix inequality $ A \succ B$ implies that matrix $ A -B $ is positive-definite. 
Multivariate Gaussian distribution with mean $ \mu $ and covariance $ \Sigma $ is denoted by $ \mathcal{N}(\mu,\Sigma) $. 
The big-$\mathcal{O}$ notation $\mathcal{O}(f(n))$ represents a function that grows at most as fast as $f(n)$.
The big-$\Omega$ notation $\Omega(f(n))$ represents a function that grows at least as fast as $f(n)$.
The big-$\Theta$ notation $\Theta(f(n))$ represents a function that grows as fast as $f(n)$.

In order to maintain the flow of the main text, the proofs of all lemmas and theorems discussed in this paper are provided in the appendix.

\section{Problem Setup}
\label{sec:problem setup}
We consider the identification problem of the MIMO  LTI system of order $n$ evolving according to
\begin{equation}\label{linear_system}
	\begin{aligned}
		x_{k+1} &= A x_{k}+ 
		w_{k}, \\
		y_{k} &= 
		C x_{k}+ v_{k},
	\end{aligned}
    \tag{\textsf{LTI-system}}
\end{equation}
based on finite output sample data, where $x_k \in \mathbb{R}^n$  and $y_k \in \mathbb{R}^m$ are the system state and the output, respectively, and the process noise $w_k \sim \mathcal{N}(\mathbf{0}, \mathcal{Q})$ and the measurement noise $v_k \sim \mathcal{N}(\mathbf{0}, \mathcal{R})$ are i.i.d. across time $k$, with $\mathcal{Q}, \mathcal{R} \succ \mathbf{0}$. The process noise and the measurement noise are mutually independent. $A \in \mathbb{R}^{n\times n}$ and $C \in \mathbb{R}^{m\times n}$  are \textbf{unknown} matrices.


\begin{assumption}\label{observable and controllable}
	\begin{enumerate}
	    \item The pair $(A,C)$ is observable, and the pair $(A, \mathcal{Q}^{1/2})$ is controllable.
        \item The system matrix $A$ is (marginally) stable, with real, distinct eigenvalues lying in the interval $[-1, 1]$.
	\end{enumerate}
\end{assumption}

\begin{remark}
	Under Assumption~\ref{observable and controllable}, the~\eqref{linear_system} is minimal: the realization $(A,C)$ has the smallest dimension among all state-space models with the same input–output relationship. Note that only the observable part of the system can be identified, and the controllability assumption ensures that every mode is excited by the process noise $w_k$. Therefore, Assumption \ref{observable and controllable} is necessary and well-defined.
\end{remark}

According to~\cite{qin2006overview}, if the pair $(A,C)$ is observable, a steady-state kalman filter can be designed for the~\eqref{linear_system} as
\begin{equation}\label{kalman filter}
	\hat{x}_{k+1} = A\hat{x} + K(y_k - C\hat{x}_k),
\end{equation}
where the steady-state kalman gain $K \in \mathbb{R}^{n \times m}$ is given by
\begin{equation}
    K  = APC^\top (CPC^\top + \mathcal{R})^{-1},
\end{equation}
and $P$ is the positive definite solution of the algebra Riccati equation
\begin{equation}
	P = APA^\top + \mathcal{Q} - APC^\top  (CPC^\top + \mathcal{R})^{-1} CP A^\top.
\end{equation}
Define the innovation of the Kalman filter as
\begin{equation}\label{e_k}
	e_k = y_k - C\hat{x}_k,
\end{equation}
one can obtain the equivalent innovation form of the~\eqref{linear_system}:
\begin{equation}\label{innovation form}
	\begin{aligned}
		\hat{x}_{k+1} &= A \hat{x}_{k}+ Ke_k, \\
		y_{k} &= C \hat{x}_{k}+ e_{k},
	\end{aligned}
    \tag{\textsf{innovation form}}
\end{equation}
where the innovation $e_k \in \mathbb{R}^m$ is the zero mean Gaussian with covariance matrix $\mathcal{S} = CPC^\top + \mathcal{R}$ and independent of past output.

\begin{assumption}\label{has converged to its steady state}
    The~\eqref{innovation form} has converged\footnote{Since the Kalman gain converges exponentially to its steady-state value, this assumption is reasonable.} to its steady state, i.e., $\hat{x}_0 \sim \mathcal{N}(\mathbf{0},P)$.
\end{assumption}

\begin{assumption}\label{initial state and stable}
	The order $n$ of the~\eqref{innovation form} is known.
\end{assumption}

The goal of this paper is to establish a \textbf{finite sample analysis} for subspace identification of $ A,C$ and $K$, up to within a similarity transformation, using multiple sample trajectories
\begin{equation}
	\left\{y_k^{(i)} \mid 0 \leq k \leq 2T-1, 1 \leq i \leq N\right\},\
\end{equation}
where \( 2T \) ($T \geq n$) and \( N \) represent the length and number of sample trajectories, respectively.



\subsection{Subspace identification method}
\label{sec:sim}

First, we define some notations. For each sample trajectory $i$, $i = 1,2,\cdots,N$, define the past outputs $y_p^{(i)} \in \mathbb{R}^{Tm}$ and future outputs $y_f^{(i)} \in \mathbb{R}^{Tm}$ as
\begin{equation}
	y_p^{(i)} = \begin{bmatrix}
		y_0^{(i)} \\ \vdots \\ y_{T-1}^{(i)}
	\end{bmatrix}, \qquad
	y_f^{(i)} = \begin{bmatrix}
		y_T^{(i)} \\ \vdots \\ y_{2T-1}^{(i)}
	\end{bmatrix}.
\end{equation}
Define the batch past outputs $Y_p \in \mathbb{R}^{Tm \times N}$ and batch future outputs $Y_f \in \mathbb{R}^{Tm \times N}$ as
\begin{equation}
	Y_p \triangleq \begin{bmatrix}
		y_p^{(1)} & \cdots & y_p^{(N)}
	\end{bmatrix}, \quad
	Y_f \triangleq \begin{bmatrix}
		y_f^{(1)} & \cdots & y_f^{(N)}
	\end{bmatrix}.
\end{equation}
Similarly, we can define the past noises $e_p^{(i)} \in \mathbb{R}^{Tm}$, the future noises $e_f^{(i)} \in \mathbb{R}^{Tm}$, the batch past noises $E_p \in \mathbb{R}^{Tm \times N}$ and the batch future noises $E_f \in \mathbb{R}^{Tm \times N}$. The batch states is defined as
\begin{equation}\label{the state}
	\widehat{X} \triangleq \begin{bmatrix}
		\hat{x}_0^{(i)} & \cdots & \hat{x}_0^{(N)}
	\end{bmatrix} \in \mathbb{R}^{n \times N}.
\end{equation}
The (extended) observability matrix $\Gamma_T \in \mathbb{R}^{Tm \times n}$, the reversed (extended) controllability matrix $\mathcal{K}_T \in \mathbb{R}^{n \times Tm}$, the block Hankel matrix $\mathcal{H}_T  \in \mathbb{R}^{Tm \times Tm}$ and the block Toeplitz matrix $\mathcal{J}_T \in \mathbb{R}^{Tm \times Tm}$ associated to the~\eqref{innovation form} are defined as
\begin{align}\label{observability matrix}
    \Gamma_T &\triangleq \begin{bmatrix}
		C^\top & (CA)^\top & \cdots & (CA^{T-1})^\top
	\end{bmatrix}^\top, \\ \label{controllability matrix}
    \mathcal{K}_T &\triangleq \begin{bmatrix}
		(A-KC)^{T-1}K  & \cdots & K
	\end{bmatrix}, \\ \label{hankel matrix}
    \mathcal{H}_T &\triangleq \Gamma_T\mathcal{K}_T, \\ \label{toeplitz}
    \mathcal{J}_T &\triangleq \begin{bmatrix}
		\mathbf{I}_m & \\
		CK & \mathbf{I}_m & \\
		\vdots & \vdots & \ddots & \\
		CA^{T-2}K & CA^{T-3}K & \cdots & \mathbf{I}_m
	\end{bmatrix}.
\end{align}


For each sample trajectory $i$, $i = 1,2,\cdots,N$, the future outputs can be written as
\begin{equation}\label{regression_1}
	y_f^{(i)} = \Gamma_T \hat{x}_T^{(i)} + \mathcal{J}_T e_f^{(i)}.
\end{equation}
On the other hand, by iterating~\eqref{kalman filter}, it is straightforward to derive that
\begin{equation}\label{regression_2}
	\hat{x}_T^{(i)} = A_C^T \hat{x}_0^{(i)} + \mathcal{K}_T y_p^{(i)}.
\end{equation}
where $A_C \triangleq A-KC$\footnote{For the sake of convenience and clarity, we will henceforth use $A_C$ to denote $A-KC$.}.
Substituting~\eqref{regression_2} into~\eqref{regression_1} and stacking all trajectories yields
\begin{equation}\label{the dynamic response}
	Y_f = \mathcal{H}_T Y_p +  \mathcal{J}_T E_f +  \Gamma_T (A-KC)^T \widehat{X}.
\end{equation}

We summarize the main steps of SIM described in~\cite{knudsen2001consistency} below.

\textbf{Step 1}: Estimate $\mathcal{H}_T$ by solving the least-squares problem:
\begin{equation}\label{regression problem}
	\widehat{\mathcal{H}}_T = \arg \min_{\mathcal{H} \in \mathbb{R}^{Tm \times Tm}} \| Y_f -  \mathcal{H} Y_p \|^2_{\rm F},
\end{equation}
with closed-form solution \( \widehat{\mathcal{H}}_T = Y_f Y_p^\dagger \), where \( Y_p^\dagger = Y_p^\top (Y_pY_p^\top)^{-1} \) is the right pseudo-inverse of \( Y_p \). Denote $\widehat{\mathcal{H}}_{T,n} \in \mathbb{R}^{Tm \times Tm}$ as the best rank-$n$ approximation of $\widehat{\mathcal{H}}_T$.


\textbf{Step 2}: Perform SVD decomposition on $\widehat{\mathcal{H}}_{T,n}$:
\begin{equation}\label{svd}
	\widehat{\mathcal{H}}_{T,n} = \mathcal U_1 S_1 \mathcal V_1^{\rm H} 
\end{equation}
where \( S_1 \in \mathbb{R}^{n \times n} \) is a diagonal matrix, with its elements being the $n$ singular values in descending order. $\mathcal{U}_1$ and $\mathcal{V}_1$ are matrices composed of corresponding left and right singular vectors, respectively. 

\textbf{Step 3}: Estimate the (extended) observability and controllability matrices:
\begin{equation}
	\widehat{\Gamma}_T = \mathcal U_1 S_1^{1/2} \mathcal{T}, \quad \widehat{\mathcal{K}}_T = \mathcal{T}^{-1} S_1^{1/2} \mathcal V_1^\top,
\end{equation}
where $\mathcal{T} \in \mathbb{R}^{n \times n}$ is an arbitrary non-singular matrix representing a similarity transformation.

\textbf{Step 4}: Recover the estimated system matrices:
\begin{equation}
	A = \widehat{\Gamma}_{T,p}^\dagger \widehat{\Gamma}_{T,f},  C = \widehat{\Gamma}_T(1:m,:),
	K = \widehat{\mathcal{K}}_T(:,-m:-1),
\end{equation}
where $ \widehat{\Gamma}_{T,p}$ and $\widehat{\Gamma}_{T,f}$ are the sub-matrices of $\widehat{\Gamma}_T$ discarding the last $m$ rows and the first $m$ rows, respectively, and $\widehat{\Gamma}_T(1:m,:)$ denotes the first $m$ rows of $\widehat{\Gamma}_T$, and $\widehat{\mathcal{K}}_T(:,-m:-1)$ denotes the last $m$ columns of $\widehat{\mathcal{K}}_T$. 

\begin{remark}
    Since the matrices $A$, $C$, and $K$ can only be identified up to a similarity transformation, we can set \( \mathcal{T} \) as the identity matrix.
\end{remark}

\section{Finite Sample Analysis of System Matrices}
\label{sec:system matrix}

In this section, we provide finite sample high-probability upper bounds for the identification error on the system matrices $A$, $C$, and $K$, based on multiple sample trajectories. Furthermore, we analyze the sample complexity required to achieve a constant estimation error.

First, we focus on analyzing the estimation error of the block Hankel matrix $\mathcal{H}_T$. 
\begin{remark}
    According to \eqref{the dynamic response} and \eqref{regression problem}, the estimation error $\widehat{\mathcal{H}}_T - \mathcal{H}_T$ is given by:
\begin{equation}\label{error_1}
	\widehat{\mathcal{H}}_T - \mathcal{H}_T = \underbrace{\mathcal{J}_TE_f Y_p^\top (Y_pY_p^\top)^{-1}}_{\text{Cross term}} + \underbrace{\Gamma_T A_C^T \widehat{X} Y_p^\top (Y_pY_p^\top)^{-1}}_{\text{Kalman filter truncation bias term}},
\end{equation}
where the first term corresponds to the cross term error, while the second term represents the truncation bias introduced by the Kalman filter. These errors arise from the noise $e_k$ and the bias in the estimated state, which results from using only $T$ past outputs instead of all the outputs.
\end{remark}

\begin{theorem}\label{estimate error of hankel matrix}
	For the~\eqref{innovation form}, under the conditions of Assumptions \ref{observable and controllable}, \ref{has converged to its steady state} and \ref{initial state and stable}, let $\widehat{\mathcal{H}}_T$ be the estimate~\eqref{regression problem} using multiple sample trajectories $ \left\{y_k^{(i)} \mid 0 \leq k \leq 2T-1, 1 \leq i \leq N\right\}$, and let $\mathcal{H}_T$ be as in Section~\ref{sec:sim}. 
    For a failure probability $\delta \in (0,1)$, if $N \geq (6+4\sqrt{2})(\sqrt{Tm} + \sqrt{2\log(1/\delta)})^2$, then with probability at least $1-5\delta$, the following upper bound holds:
	\begin{equation}\label{upper bound of estimation error of hankel matrix}
		\|\widehat{\mathcal{H}}_T - \mathcal{H}_T\| \leq \frac{2}{ \lambda_{\min}(\mathcal{R})}  \left( \mathcal{C}_1 \mathcal{C}_2 + \sqrt{N}\|A_C^T\| \mathcal{C}_3 \right) \sqrt{\frac{T^5}{N}},
	\end{equation}
	where $\overline{c} \triangleq \max_{i,j} |c_{ij}|$ and $\overline{k} \triangleq \max_{i,j} |k_{ij}|$ denote the maximum absolute values of the entries of matrices $C$ and $K$, respectively. The constants $\mathcal{C}_1$, $\mathcal{C}_2$ and $\mathcal{C}_3$ are given by
    \begin{align} \label{c_1}
        \mathcal{C}_1 &= 4 \sqrt{2\log (9/\delta)} m^{3/2}(1+\overline{c}\overline{k}n) \|\mathcal{S}\|^{1/2},\\ \label{c_2}
        \mathcal{C}_2 &=  \sqrt{\frac{mn}{T}}\overline{c} \| P \|^{1/2}  + m(1+\overline{c}\overline{k}n)\| \mathcal{S} \|^{1/2}, \\ \label{c_3}
        \mathcal{C}_3 &= \sqrt{mn}\overline{c}\left(\frac{\sqrt{3}}{T^2} \| P \|  + \frac{\mathcal{C}_1}{\sqrt{TN}}  \| P \|^{1/2} \right).
    \end{align}
\end{theorem}

Note that the spectral radius of the matrix $A_C = A - KC$ is strictly less than one. As a result, the second term in~\eqref{upper bound of estimation error of hankel matrix}, which corresponds to the truncation bias introduced by the Kalman filter, decays exponentially with respect to $T$. In this regard,the leading contribution to the estimation error stems from the first term, i.e., the cross term error.

To balance the decay rates of the cross term and the truncation bias term, we set $T \sim \Theta (\log N)$. Under this setting, the overall estimation error of the block Hankel matrix $\mathcal{H}_T$ decays at the rate of at least $\mathcal{O}\left( \frac{(\log N)^{5/2}}{\sqrt{N}} \right)$.

\begin{remark}
	It is worth noting that in the absence of input, the noise both helps and hinders identification. The larger the noise, the better the excitation effect of the output, but it also reduces the convergence of the estimator of the block Hankel matrix $\mathcal{H}_T$. To understand how the non-asymptotic bound~\eqref{upper bound of estimation error of hankel matrix} captures this, observe that as the noise increases, $\|\mathcal{S}\|$ becomes larger, but $\lambda_{\min}(\mathcal{R})$ also increases. 
\end{remark}


Second, we shows the robustness of the SIM to adversarial disturbances that may occur in the block Hankel matrix $\mathcal{H}_T$. 

\begin{lemma}\label{the robustness of balanced realizations}
	For the~\eqref{innovation form}, under the conditions of Assumptions \ref{observable and controllable}, \ref{has converged to its steady state} and \ref{initial state and stable}, let $\overline{A}, \overline{C}$, $\overline{K}$ be the state-space realization corresponding to the output of SIM with $\mathcal{H}_T$ and $\widehat{A}, \widehat{C}$, $\widehat{K}$ be the state-space realization corresponding to the output of SIM with $\widehat{\mathcal{H}}_T$.
    Suppose $\sigma_{n}(\mathcal{H}_T) > 0$ and perturbation obeys
	\begin{equation}\label{perturbation condition}
		\|\widehat{\mathcal{H}}_T - \mathcal{H}_T\|  \leq \frac{\sigma_{n}(\mathcal{H}_T)}{4},
	\end{equation}
	then there exists a unitary matrix $\mathcal{U} \in \mathbb{R}^{n \times n}$ such that
    	\begin{multline}
		\max\left\{\| \widehat{C} - \overline{C} \mathcal{U} \| ,	\| \widehat{K} -\mathcal{U}^\top \overline{K} \| \right\} \\
        \leq \sqrt{\frac{39n}{\sigma_{n}(\mathcal{H}_T)}} \| \mathcal{H}_T - \widehat{\mathcal{H}}_{T} \|.
	\end{multline}
	Furthermore, the matrices $\overline{A}$, $\widehat{A}$ satisfy
	\begin{equation}\label{the upper bound of A}
		\| \widehat{A} - \mathcal{U}^\top \overline{A} \mathcal{U} \|  \leq \left(\sqrt{2} + 2 \sqrt{
        \frac{\overline{c}\overline{k}Tmn}{\sigma_{n}(\mathcal{H}_T)}
        }\right) \frac{\sqrt{39n}}{\sigma_{n}(\mathcal{H}_T)}\|\mathcal{H}_T - \widehat{\mathcal{H}}_{T} \|.
	\end{equation}
\end{lemma}

\begin{remark}
    Lemma~\ref{the robustness of balanced realizations} reveals that the state-space realization of SIM is robust to noise up to trivial ambiguities. The robustness is controlled by the smallest singular value $\sigma_{n}(\mathcal{H}_T)$ of the block Hankel matrix $\mathcal{H}_T$. Here ``weakest'' is in terms of observability and controllability, therefore, a smaller $\sigma_{n}(\mathcal{H}_T)$ implies that there is a mode of the~\eqref{innovation form} is more difficult to identify.
\end{remark}

Finally, our next result combines the robustness of SIM with the finite sample identification error bound in Theorem~\ref{estimate error of hankel matrix} to obtain end-to-end estimation guarantees for the matrices $A,C,K$.

\begin{theorem}\label{end-to-end}
	For the~\eqref{innovation form}, under the conditions of Assumptions \ref{observable and controllable}, \ref{has converged to its steady state} and \ref{initial state and stable}, let $\widehat{\mathcal{H}}_T$ be the estimate~\eqref{regression problem} using multiple sample trajectories $ \left\{y_k^{(i)} \mid 0 \leq k \leq 2T-1, 1 \leq i \leq N\right\}$, and let $\mathcal{H}_T$ be as in Section~\ref{sec:sim}. Let $\overline{A}, \overline{C}$, $\overline{K}$ be the state-space realization corresponding to the output of SIM with $\mathcal{H}_T$ and $\widehat{A}, \widehat{C}$, $\widehat{K}$ be the state-space realization corresponding to the output of SIM with $\widehat{\mathcal{H}}_T$. For a failure probability $\delta \in (0,1)$, if the number of sample trajectories $N$ satisfies that $N \geq (6+4\sqrt{2})(\sqrt{Tm} + \sqrt{2\log(1/\delta)})^2$, and $\sigma_{n}(\mathcal{H}_T) > 0$, and perturbation obeys
	\begin{equation}
		\|\widehat{\mathcal{H}}_T - \mathcal{H}_T\|  \leq \frac{\sigma_{n}(\mathcal{H}_T)}{4},
	\end{equation}
	then there exists a unitary matrix $\mathcal{U} \in \mathbb{R}^{n \times n}$ such that with probability at least $1-6\delta$,
	\begin{multline}
		\max\left\{\| \widehat{C} - \overline{C} \mathcal{U} \| ,	\| \widehat{K} -\mathcal{U}^\top \overline{K} \| \right\} \\
        \leq  \frac{2}{ \lambda_{\min}(\mathcal{R})}  \left( \mathcal{C}_1 \mathcal{C}_2 + \sqrt{N}\|A_C^T\| \mathcal{C}_3 \right) \sqrt{\frac{39nT^5}{N \sigma_{n}(\mathcal{H}_T)}},
	\end{multline}
	and
	\begin{multline}
		\max\left\{\| \widehat{A} - \mathcal{U}^\top \overline{A} \mathcal{U} \| \right\} 
		\leq \frac{2}{ \lambda_{\min}(\mathcal{R})}  \left( \mathcal{C}_1 \mathcal{C}_2 + \sqrt{N}\|A_C^T\| \mathcal{C}_3 \right) \\
        \cdot \left(\sqrt{2} + 2 \sqrt{
        \frac{\overline{c}\overline{k}Tmn}{\sigma_{n}(\mathcal{H}_T)}
        } \right) \frac{\sqrt{39nT^5}}{\sqrt{N}\sigma_{n}(\mathcal{H}_T)},
	\end{multline}
	where the constants $\mathcal{C}_1$, $\mathcal{C}_2$ and $\mathcal{C}_3$ are respectively defined in~\eqref{c_1},~\eqref{c_2} and~\eqref{c_3}.
\end{theorem}

According to Theorem~\ref{end-to-end}, $\sigma_{n}(\mathcal{H}_T)$ has a non-negligible impact on the upper bound of the identification error of the system matrices $A$, $C$, and $K$. To further elucidate this influence, the following lemma offers a rigorous theoretical characterization of $\sigma_{n}(\mathcal{H}_T)$.

\begin{lemma}\label{ill-conditioned}
	For the~\eqref{innovation form}, under the conditions of Assumption \ref{observable and controllable}, then the $n$-th largest singular value of $ \mathcal{H}_T $ satisfies the following inequality
	\begin{equation}\label{L}
		\sigma_{n}(\mathcal{H}_T) \leq 4\overline{c}\overline{k} Tmn \rho^{-\frac{\left\lfloor \frac{n-1}{2m} \right\rfloor}{\log (2mT)} },
	\end{equation}
	where $ \rho \triangleq e^{\frac{\pi^2}{4}} \approx 11.79$, and $\overline{c} \triangleq \max_{i,j} |c_{ij}|$ and $\overline{k} \triangleq \max_{i,j} |k_{ij}|$ denote the maximum absolute values of the entries of matrices $C$ and $K$, respectively.
\end{lemma}

\begin{remark}
    Lemma~\ref{ill-conditioned} reveals that under the given conditions above, $\sigma_{n}(\mathcal{H}_T)$ decays super-polynomially with respect to $n/m$.
\end{remark}

For the~\eqref{innovation form}, the results of Theorem~\ref{end-to-end} and Lemma~\ref{ill-conditioned} reveal that, if the sample parameter satisfies $T \sim \Theta(\log N)$, the estimation errors of matrices $C,K$ decay at the rate of at least $\mathcal{O}\left( \frac{(\log N)^{2}}{\sqrt{N}} \right)$, while the estimation error of matrix $A$ decays at the rate of at least $\mathcal{O}\left( \frac{(\log N)^{3/2}}{\sqrt{N}} \right)$.

Moreover, the required number of sample trajectories \( N \) to achieve a constant estimation error for the matrices $C,K$ with high probability must satisfy
\begin{equation}
    N \sim \Omega \left( (mn \log N)^4 \rho^{ \frac{\left\lfloor \frac{n-1}{2m} \right\rfloor}{\log (2m\log N)}} 
    \right),
\end{equation}
where \( \rho \triangleq  e^{\frac{\pi^2}{4}} \approx 11.79 \). Similarly, the required number of sample trajectories \( N \) to achieve a constant estimation error for the matrix $A$ with high probability must satisfy
\begin{equation}\label{appro_A}
    N \sim \Omega \left( (mn \log N)^3 \varrho^{ \frac{\left\lfloor \frac{n-1}{2m} \right\rfloor}{\log (2m\log N)}} 
    \right),
\end{equation}
where \( \varrho \triangleq \rho^3 = e^{\frac{3\pi^2}{4}} \approx 1639.59 \). 

These results imply that, ignoring the logarithmic factors, the estimation errors of matrices $A,C,K$ decay at the rate of at least $\mathcal{O}\left(1/\sqrt{N}\right)$, and achieving a constant estimation error requires a \textbf{super-polynomial} sample size $N \times T$ in $n/m$, where $n/m$ denotes the state-to-output dimension ratio.

\section{Finite Sample Analysis of System Poles}
\label{sec:system pole}

To date, our investigation has focused on characterizing the finite sample identification error of the system matrices. Since poles are invariant under similarity transformations and play a key role in many controller-design methods, it is also important to analyze their finite sample identification error. Building on the previous results, this section further derives finite sample high-probability upper bounds on the identification error of the system poles.

To characterize the gap between the spectra of the matrices $\overline{A}$ and $\widehat{A}$, we introduce the Hausdorff distance~\cite{hausdorff1914grundzuge}, definied as follows.

\begin{definition}[Hausdorff Distance~\cite{hausdorff1914grundzuge}]
	Given $\mathcal{A} = (\alpha_{ij}) \in \mathbb{C}^{n \times n}$ and $\mathcal{B} = (\beta_{ij}) \in \mathbb{C}^{n \times n}$, suppose that $\lambda(\mathcal{A}) =  \{\lambda_1(\mathcal{A}), \cdots, \lambda_n(\mathcal{A})\}$ and $\lambda(\mathcal{B}) = \{\mu_1(\mathcal{B}), \cdots, \mu_n(\mathcal{B})\}$ are the spectra of matrix $\mathcal{A}$ and $\mathcal{B}$, respectively, then the Hausdorff distance between the spectra of matrix $\mathcal{A}$ and $\mathcal{B}$ is defined as
	\begin{equation}
		d_{\rm H}(\mathcal{A},\mathcal{B}) \triangleq \max\{{\rm sv}_{\mathcal{A}}(\mathcal{B}), {\rm sv}_{\mathcal{B}}(\mathcal{A})\},
	\end{equation}
	where the spectrum variation of $\mathcal{B}$ with respect to $\mathcal{A}$ is given by
	\begin{equation}
		{\rm sv}_{\mathcal{A}}(\mathcal{B}) \triangleq \max_{1\leq j\leq n} \min_{1\leq i\leq n} | \lambda_i(\mathcal{A})- \mu_j(\mathcal{B})|.
	\end{equation} 
\end{definition}


The following result demonstrates the robustness of the system poles for the SIM against adversarial disturbances that may arise in the block Hankel matrix $\mathcal{H}_T$. This result follows the steps in~\cite{sun2024nonasymptoticerroranalysissubspace}.

\begin{theorem}\label{the system poles}
	For the~\eqref{innovation form}, under the conditions of Assumptions \ref{observable and controllable}, \ref{has converged to its steady state} and \ref{initial state and stable}, let $\widehat{\mathcal{H}}_T$ be the estimate~\eqref{regression problem} using multiple sample trajectories $ \left\{y_k^{(i)} \mid 0 \leq k \leq 2T-1, 1 \leq i \leq N\right\}$, and let $\mathcal{H}_T$ be as in Section~\ref{sec:sim}. Let $\overline{A}, \overline{C}$, $\overline{K}$ be the state-space realization corresponding to the output of SIM with $\mathcal{H}_T$ and $\widehat{A}, \widehat{C}$, $\widehat{K}$ be the state-space realization corresponding to the output of SIM with $\widehat{\mathcal{H}}_T$. For a failure probability $\delta \in (0,1)$, if the number of sample trajectories $N$ satisfies that $N \geq (6+4\sqrt{2})(\sqrt{Tm} + \sqrt{2\log(1/\delta)})^2$, and $\sigma_{n}(\mathcal{H}_T) > 0$, and perturbation obeys
	\begin{equation}
		\|\widehat{\mathcal{H}}_T - \mathcal{H}_T\|  \leq \frac{\sigma_{n}(\mathcal{H}_T)}{4},
	\end{equation}
	then with probability at least $1-6\delta$,
	\begin{equation}\label{upper bound for system pole}
		d_{\rm H}(\widehat{A},\overline{A}) \leq 
		 (\Delta + 2\|\overline{A}\|)^{1-\frac{1}{n}} \Delta^{\frac{1}{n}}
	\end{equation}
	where 
    \begin{equation}
        \begin{aligned}
            \Delta &= \frac{2}{ \lambda_{\min}(\mathcal{R})}  \left( \mathcal{C}_1 \mathcal{C}_2 + \sqrt{N}\|A_C^T\| \mathcal{C}_3 \right)  \\
         & \qquad \qquad \cdot \left(\sqrt{2} + 2 \sqrt{
        \frac{\overline{c}\overline{k}Tmn}{\sigma_{n}(\mathcal{H}_T)}
        } \right) \frac{\sqrt{39nT^5}}{\sqrt{N}\sigma_{n}(\mathcal{H}_T)},
        \end{aligned}
    \end{equation}
    and the constants $\mathcal{C}_1$, $\mathcal{C}_2$ and $\mathcal{C}_3$ are respectively defined in~\eqref{c_1},~\eqref{c_2} and~\eqref{c_3}.
\end{theorem}

\begin{remark}
    Theorem~\ref{the system poles} provides end-to-end estimation guarantees for the system poles and shows that the estimation error of the system poles is controlled by that of the matrix \( A \). This is not surprising, as the eigenvalues of a matrix exhibit good continuity. 
\end{remark}

For the~\eqref{innovation form}, the results of Theorem~\ref{the system poles} and Lemma~\ref{ill-conditioned} reveal that, ignoring the logarithmic factors, if the sample parameter satisfies $T \sim \Theta(\log N)$, the estimation error of the system poles decays at the rate of at least $\mathcal{O}\left( N^{-1/2n} \right)$. Furthermore, in light of~\eqref{appro_A} and~\eqref{upper bound for system pole}, achieving a constant estimation error also requires a \textbf{super-polynomial} sample size $N \times T$ in $n/m$, where $n/m$ denotes the state-to-output dimension ratio.

\section{Numerical Results}
\label{sec:numerical results}

In this section, we first investigate the identification of the classical two-mass spring-damper system \cite{bhutta2016interactive}, using the SIM. The system consists of two point masses, $m_1 = 1$ kg and $m_2 = 1$ kg, interconnected in series via linear springs and viscous dampers. Each mass is also coupled to a fixed boundary via a spring-damper pair, as illustrated in \autoref{fig:enter-label}.

\begin{figure}[H]
    \centering
    \begin{tikzpicture}[
       spring/.style={  
        thick,  
        decoration={  
            coil,  
            aspect=0.5,  
            segment length=1mm,  
            amplitude=1mm,  
            pre length=3mm,  
            post length=3mm  
        },
        decorate  
    },
   wall/.style={  
        thick,
        pattern=north east lines,  
        minimum height=1.5cm,  
        minimum width=1pt  
    },
    wheel/.style={  
        circle,
        fill=white,
        draw=black,
        thick,
        minimum size=4pt,  
        outer sep=0pt
    },
    ground/.style={  
        thick,
        pattern=north east lines,  
        minimum height=0.1cm,  
        anchor=north  
    },
     damper/.style={
        draw, 
        thick,
        line width=1pt,
        minimum width=0.3cm,
        minimum height=0.2cm,
        inner sep=0pt,
        outer sep=0pt
    }
]

\def\wallsep{7cm}     
\def\masswidth{1.2cm} 
\def\massheight{0.8cm} 
\def\groundlevel{-1cm}  
\def\wheeloffset{0.3cm} 
\def\dampershift{0.3cm} 

\node (ground) [ground, minimum width=\wallsep, anchor=north] at (\wallsep/2, \groundlevel-5) {};
\draw[thick] (ground.north west) -- (ground.north east);

\node (left wall) [wall, anchor=east] at (0, \groundlevel+16.5) {};  
\draw[thick] (left wall.north east) -- (left wall.south east);  

\node (right wall) [wall, anchor=west] at (\wallsep, \groundlevel+16.5) {};  
\draw[thick] (right wall.north west) -- (right wall.south west);  

\node (mass1) [draw, thick, minimum width=\masswidth, minimum height=\massheight, anchor=south] at (2, \groundlevel) {$m_1$};
\fill (1.4,\massheight/2+\groundlevel+6) circle (1pt);  
\fill (2.6,\massheight/2+\groundlevel+6) circle (1pt);  
\fill (1.4,\massheight/2+\groundlevel-6.5) circle (1pt); 
\fill (2.6,\massheight/2+\groundlevel-6.5) circle (1pt); 

\fill (0,\massheight/2+\groundlevel+6) circle (1pt); 
\fill (0,\massheight/2+\groundlevel-6.5) circle (1pt); 
\fill (7,\massheight/2+\groundlevel+6) circle (1pt); 

\node [wheel] at ([xshift=-\wheeloffset]mass1.south) {};  
\node [wheel] at ([xshift=\wheeloffset]mass1.south) {};   

\node (mass2) [draw, thick, minimum width=\masswidth, minimum height=\massheight, anchor=south] at (5, \groundlevel) {$m_2$};
\fill (4.4,\massheight/2+\groundlevel+6) circle (1pt);  
\fill (5.6,\massheight/2+\groundlevel+6) circle (1pt);  
\fill (4.4,\massheight/2+\groundlevel-6.5) circle (1pt); 

\node [wheel] at ([xshift=-\wheeloffset]mass2.south) {};  
\node [wheel] at ([xshift=\wheeloffset]mass2.south) {};   

\draw[spring] (0,\massheight/2+\groundlevel+6) -- (1.4,\massheight/2+\groundlevel+6)
    node [midway, above=0.5mm] {\small $k_1$};

\draw[spring] (2.6,\massheight/2+\groundlevel+6) -- (4.4,\massheight/2+\groundlevel+6)
    node [midway, above=0.5mm] {\small $k_2$};

\draw[spring] (5.6,\massheight/2+\groundlevel+6) -- (7,\massheight/2+\groundlevel+6)
    node [midway, above=0.5mm] {\small $k_3$};

\coordinate (dampStart1) at ([yshift=-\dampershift]0,\massheight/2+\groundlevel+2);
\coordinate (dampEnd1) at ([yshift=-\dampershift]1.4,\massheight/2+\groundlevel+2);

\node [damper, anchor=center] (d1) at ($(dampStart1)!0.5!(dampEnd1)$) {};

\draw[thick] (dampStart1) -- (d1.west);
\draw[thick] (d1.east) -- (dampEnd1);

\draw[thick] (d1.north east) -- ++(0.1cm,0);
\draw[thick] (d1.south east) -- ++(0.1cm,0);

\node[below] at (1.1,-0.8) {\small $c_1$};

\coordinate (dampStart2) at ([yshift=-\dampershift]2.6,\massheight/2+\groundlevel+2);
\coordinate (dampEnd2) at ([yshift=-\dampershift]4.4,\massheight/2+\groundlevel+2);

\node [damper, anchor=center] (d2) at ($(dampStart2)!0.5!(dampEnd2)$) {};

\draw[thick] (dampStart2) -- (d2.west);
\draw[thick] (d2.east) -- (dampEnd2);

\draw[thick] (d2.north east) -- ++(0.1cm,0);

\draw[thick] (d2.south east) -- ++(0.1cm,0);

\node[below] at (1.1+2.8,-0.8) {\small $c_2$};

\draw[thick, ->,black] (2,-0.2+0.2) -- ++(0.8,0) 
    node[above] {$q_1$};
\draw[thick, ->, black] (5,-0.2+0.2) -- ++(0.8,0) 
    node[above] {$q_2$};

\draw[thick] (2,-0.2) -- (2,-0.2+0.6);

\draw[thick] (5,-0.2) -- (5,-0.2+0.6
);

\end{tikzpicture}
    \caption{A two-mass spring-damper system.}
    \label{fig:enter-label}
\end{figure}

Let $q_1(t)$ and $q_2(t)$ denote the horizontal displacements of the masses $m_1$ and $m_2$, respectively. These displacements constitute the system outputs. Applying Newton's second law yields the following equations of motion:
\begin{align}\notag
    m_1 \ddot{q}_1 &= -k_1 q_1 + k_2(q_2 - q_1) - c_1 \dot{q}_1 + c_2(\dot{q}_2 - \dot{q}_1), \\\notag
    m_2 \ddot{q}_2 &= -k_2(q_2 - q_1) - k_3 q_2 - c_2(\dot{q}_2 - \dot{q}_1).
\end{align}
Defining the state vector and output vector as $x(t) = \begin{bmatrix}
        q_1(t) & \dot{q}_1(t) & q_2(t) & \dot{q}_2(t)
    \end{bmatrix}^\top$,  $y(t) = \begin{bmatrix}
        q_1(t)  & q_2(t) 
    \end{bmatrix}^\top$,
the system dynamics can be expressed in the continuous-time state-space form:
\begin{align}\notag
    \dot{x}(t) &= A_cx(t), \\\notag
    y(t) &= C_cx(t),
\end{align}
where 
$
A_c = \begin{bmatrix}
0 & 1 & 0 & 0 \\
-\frac{k_1 + k_2}{m_1} & -\frac{c_1 + c_2}{m_1} & \frac{k_2}{m_1} & \frac{c_2}{m_1} \\
0 & 0 & 0 & 1 \\
\frac{k_2}{m_2} & \frac{c_2}{m_2} & -\frac{k_2 + k_3}{m_2} & -\frac{c_2}{m_2}
\end{bmatrix}, \quad
C_c = \begin{bmatrix}
1 & 0 & 0 & 0 \\
0 & 0 & 1 & 0
\end{bmatrix}
$.

To facilitate identification and digital implementation, the continuous-time system is discretized via zero-order hold (ZOH) with a sampling interval of $T_s = 0.1$ seconds. The resulting discrete-time  stochastic model is:
\begin{equation}\label{discrete-time system}
    \begin{aligned}
        x_{k+1} &= A_d x_k + w_k, \\
        y_k &= C_d x_k + v_k,
    \end{aligned}
    \tag{\textsf{Test-system}}
\end{equation}
where $A_d = e^{A_c T_s}$, $C_d = C_c$. The process noise $w_k \sim \mathcal{N}(\mathbf{0}, 10^{-4}\mathbf{I}_4)$ models the environmental vibrations acting on the structure, while the measurement noise $v_k \sim \mathcal{N}(\mathbf{0}, 10^{-4}\mathbf{I}_2)$ accounts for sensor errors; both noise sequences are assumed to be i.i.d. and mutually independent.

For the~\eqref{discrete-time system}, we consider two different parameter sets corresponding to distinct stability categories—stable and marginally stable.
\begin{itemize}
    \item Stable: $k_1 = 0.5$ N/m, $k_2 = 0.7$ N/m, $k_3 = 0.6$ N/m; $c_1 = c_2 = 5$ N$\cdot$s/m. Matrix $A_d$ has the spectrum $\lambda(A_d) = \{0.27,0.99,0.95,0.86\}$.
    \item Marginally stable: $k_1 = 0.5$ N/m, $k_2 = 0.7$ N/m, $k_3 = 0.6$ N/m; $c_1 = 60$ N$\cdot$s/m, $c_2 = 5$ N$\cdot$s/m. Matrix $A_d$ has the spectrum $\lambda(A_d) = \{0.001,0.65,0.97,1.00\}$.
\end{itemize}

In the experiments, the number of trajectories, $N$, varies from $500$ to $5000$, while the length of sample trajectories $T$ is set to $\lceil \log N \rceil$. Each scenario is evaluated over $100$ independent trials. 

As shown in~\autoref{fig:all}, \textbf{\romannumeral1}) subfigures (a) and (b) illustrate the estimation error of the block Hankel matrix $\mathcal{H}_T$ for stable (red) and marginally stable (blue) systems, respectively, under varying numbers of sample trajectories $N$; \textbf{\romannumeral2}) subfigures (b) and (e) depict the estimation error of the system poles for stable (red) and marginally stable (blue) systems, respectively, under different values of $N$, where the estimation error of the system poles is measured using the Haussdorff distance; \textbf{\romannumeral3}) subfigures (c) and (f) compare the distribution of true poles (red `$\circ$') with the distribution of estimated poles (blue `+') for stable and marginally stable systems, respectively. The results shown in subfigures (a), (b), (d), and (e) are visualized using shaded error bar.

According to~\autoref{fig:H_stable} and~\autoref{fig:H_marginal}, the estimation error of the block Hankel matrix $\mathcal{H}_T$ decreases as the number of sample trajectories $N$ increases, both for stable and marginally stable systems. In contrast, \autoref{fig:pole_stable} and~\autoref{fig:pole_marginal} indicate that the estimation error of the system poles decreases at a slower rate as $N$ increases.

\begin{figure*}[bthp]
    \centering
    \begin{subfigure}[t]{0.32\linewidth}
        \centering
        \includegraphics[width=\linewidth]{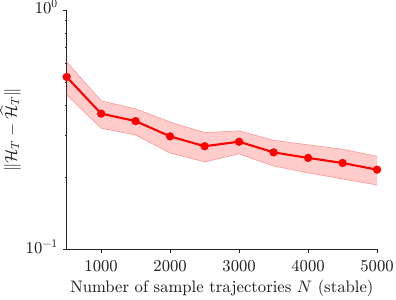}
        \caption{}
        \label{fig:H_stable}
    \end{subfigure}%
    \hfill
    \begin{subfigure}[t]{0.32\linewidth}
        \centering
        \includegraphics[width=\linewidth]{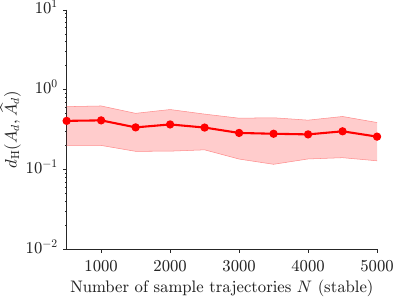}
        \caption{}
        \label{fig:pole_stable}
    \end{subfigure}%
    \hfill
    \begin{subfigure}[t]{0.25\linewidth}
        \centering
        \includegraphics[width=\linewidth]{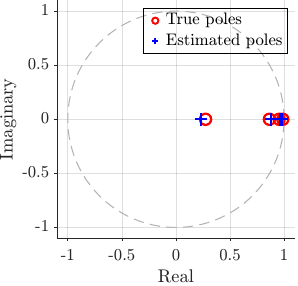}
        \caption{}
        \label{fig:pole_distribution_stable}
    \end{subfigure}
    \hfill

     \begin{subfigure}[t]{0.32\linewidth}
        \centering
        \includegraphics[width=\linewidth]{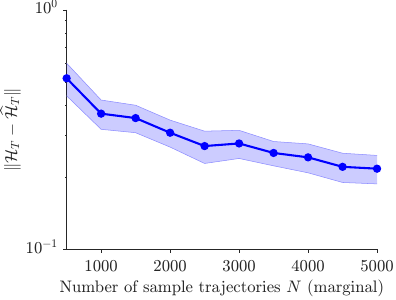}
        \caption{}
        \label{fig:H_marginal}
    \end{subfigure}%
    \hfill
    \begin{subfigure}[t]{0.32\linewidth}
        \centering
        \includegraphics[width=\linewidth]{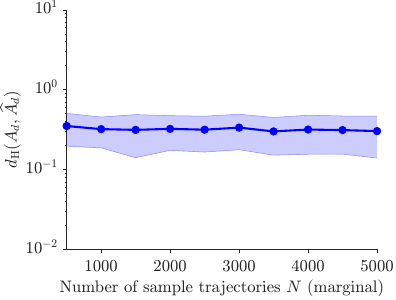}
        \caption{}
        \label{fig:pole_marginal}
    \end{subfigure}%
    \hfill
    \begin{subfigure}[t]{0.25\linewidth}
        \centering
        \includegraphics[width=\linewidth]{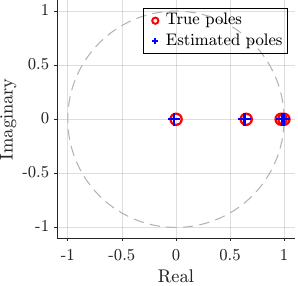}
        \caption{}
        \label{fig:pole_distribution_marginal}
    \end{subfigure}
    \caption{Subfigures (a) and (d): The estimation error of the block Hankel matrix $\mathcal{H}_T$ for stable (red) and marginally stable (blue) systems, respectively, under varying numbers of sample trajectories $N$. Subfigures (b) and (e): The estimation error of the system poles for stable (red) and marginally stable (blue) systems, respectively, under different values of $N$. Subfigures (c) and (f): Comparison of the distribution of true poles (red `$\circ$') with the distribution of estimated poles (blue `+') for stable and marginally stable systems, respectively.}
    \label{fig:all}
\end{figure*}

\section{Conclusions}
\label{sec:conclusions}

This paper presents an error analysis for the identification of discrete-time LTI systems using the SIM with finite sample data. We derive high-probability error bounds for the system matrices $A,C$, the Kalman filter gain $K$ and the estimation of system poles. Specifically, we prove that for an $n$-dimensional LTI system with no external inputs, the estimation error of these matrices decreases at a rate of at least $ \mathcal{O}(\sqrt{1/N}) $ with $N$ independent sample trajectories, while the estimation error of the system poles decays at a rate of at least $ \mathcal{O}(N^{-1/2n}) $. Additionally, we demonstrate that achieving a constant estimation error requires a super-polynomial sample size in $n/m $, where $n/m$ denotes the state-to-output dimension ratio. In future work, we will explore extending this analysis to systems with external inputs.

\section*{Appendix}

\subsection{Some Lemmas}

\begin{lemma}[Two-sided bound on Gaussian matrices, \cite{rudelson2009smallest}]\label{Two-sided bound on Gaussian matrices}
    Let $M$ be an $N \times n$ ($N \geq n$) matrix whose entries are independently and identically distributed Gaussian with mean zero and unit. Then for any $ t \geq 0$, we have
    \begin{equation}
        \mathbb{P}\left[
            \sigma_{\min}(A) \leq \sqrt{N} - \sqrt{n} - t
        \right] \leq e^{-\frac{t^2}{2}},
    \end{equation}
    and
    \begin{equation}
        \mathbb{P}\left[
            \|A\| \geq \sqrt{N} + \sqrt{n} + t
        \right] \leq e^{-\frac{t^2}{2}}.
    \end{equation}
\end{lemma}

\begin{corollary}\label{two-sided bound}
    Let $x_i \in \mathbb{R}^n$ be independently and identically distributed Gaussian with mean zero and covariance matrix $\Sigma_x$, and set $M=\sum_{i=1}^N x_i x_i^{\top}$. Fix a failure probability $\delta \in(0,1]$, and let $N \geq (6+4\sqrt{2})(\sqrt{n}+\sqrt{2\log(1/\delta)})^2$. Then with probability at least $1-\delta$, we have that 
    \begin{equation}\label{lower bound for gaussian matrix}
        \lambda_{\min }(M) \geq \frac{N}{2} \cdot \lambda_{\min }\left(\Sigma_x\right).
    \end{equation}
    On the other hand, with probability at least $1-\delta$, we have that 
    \begin{equation}\label{upper bound for gaussian matrix}
        \| M \| \leq \sqrt{3} N \| \Sigma_x\| .
    \end{equation}
\end{corollary}

\begin{proof}
    We can rewrite $x_i = \Sigma_x^{1/2}y_i$, where $y_i \in \mathbb{R}^n$ is Gaussian with mean zero and unit covariance. Thus, we have $M = \Sigma_x^{1/2} \left(\sum_{i=1}^N y_i y_i^{\top}\right) \Sigma_x^{1/2}$, then it follows that 
    \begin{equation}\label{corollary-1}
        \lambda_{\min }(M) \geq \lambda_{\min }\left(\Sigma_x\right) \lambda_{\min } \left( \sum_{i=1}^N y_i y_i^{\top}\right)
    \end{equation}
    and
    \begin{equation}\label{corollary-2}
        \| M \| \leq \| \Sigma_x\| \left\| \sum_{i=1}^N y_i y_i^{\top} \right\|.
    \end{equation}
    Applying Lemma~\ref{Two-sided bound on Gaussian matrices}, with probability at least $1-\delta$, we have that 
    \begin{equation}\label{corollary-3}
        \lambda_{\min } \left( \sum_{i=1}^N y_i y_i^{\top}\right) \geq 
        \left( \sqrt{N} - \sqrt{n} - \sqrt{2\log(1/\delta)} \right)^2 \geq \frac{N}{2}.
    \end{equation}
    Combing \eqref{corollary-1} with \eqref{corollary-3} completes the proof of \eqref{lower bound for gaussian matrix}.
    Applying Lemma~\ref{Two-sided bound on Gaussian matrices} again, with probability at least $1-\delta$, we have that 
    \begin{multline}\label{corollary-4}
        \left\| \sum_{i=1}^N y_i y_i^{\top} \right\| \leq \left( \sqrt{N} + \sqrt{n} +\sqrt{2\log(1/\delta)} \right)^2 \\
        \leq \left(2-\frac{1}{\sqrt{2}}\right)^2N < \sqrt{3}N.
    \end{multline}
    Combing \eqref{corollary-2} with \eqref{corollary-4} completes the proof of \eqref{upper bound for gaussian matrix}.
\end{proof}

\begin{lemma}[Proposition III.1,\cite{matni2019tutorial}]\label{xw}
	Let $x_i \in \mathbb{R}^n$ and $w_i \in \mathbb{R}^m$ be such that $x_i \in \mathbb{R}^n$ be independently and identically distributed Gaussian with mean zero and covariance matrix $\Sigma_x$, and $w_i \in \mathbb{R}^n$ be independently and identically distributed Gaussian with mean zero and covariance matrix $\Sigma_w$, and let $M=$ $\sum_{i=1}^N x_i w_i^{\top}$. Fix a failure probability $\delta \in(0,1]$, and let $N \geq$ $\frac{1}{2}(n+m) \log (9 / \delta)$. Then, it holds with probability at least $1-\delta$ that
	\begin{equation}
		\|M\| \leq 4\left\|\Sigma_x\right\|^{1 / 2}\left\|\Sigma_w\right\|^{1 / 2} \sqrt{N(n+m) \log (9 / \delta)}.
	\end{equation}
\end{lemma}

\begin{lemma}[Lemma 7.4,\cite{oymak2021revisiting}]\label{sigma}
	For the true block Hankel matrix $\mathcal{H}_T$ defined in Section~\ref{sec:sim} and the noisy estimate $\widehat{\mathcal{H}}_T$ defined in~\eqref{regression problem}, suppose $\sigma_{\min}(\mathcal{H}_T) \geq 2 \| \mathcal{H}_T - \widehat{\mathcal{H}}_T \|$. Then, $\|  \widehat{\mathcal{H}}_T  \| \leq 2\| \widehat{H}_T \|$ and $\sigma_{\min}(\widehat{\mathcal{H}}_T) \geq \sigma_{\min}(\mathcal{H}_T)/2$.
\end{lemma}

%

\begin{lemma}[Lemma 4, \cite{sun2023finite}]\label{krylov matrix}
	Given an $ n \times mp $ ($ p \leq  n $) matrix $ X_{n,mp} $ satisfying 
	\begin{equation}\notag
		X_{n,mp} = \begin{bmatrix}
			J_p &
			DJ_p &
			\cdots &
			D^{m-1}J_p
		\end{bmatrix},
	\end{equation}	
	where $ J_p \in \mathbb{R}^{n \times p} $ and $ D $ is a unitary diagonalizable matrix with real eigenvalues, the smallest singular value of $ X_{n,mp} $ satisfies
	\begin{equation}\notag
		\sigma_{\min}(X_{n,mp}) \leq 4 \rho_{}^{-\frac{\left\lfloor \frac{\min\{n,m(p-[p]_*)\}-1}{2p} \right\rfloor}{\log (2mp)}} \|X_{n,mp}\|, 
	\end{equation}
	where $ \rho \triangleq e^{\frac{\pi^2}{4}} $, and $ [p]_* = 0 $ if $ p $ is even or $ p = 1 $ and is $ 1 $ if $ p $ is an odd number greater than $ 1 $.
\end{lemma}

\begin{theorem}[Theorem 1, \cite{elsner1985optimal}]\label{perturbation of eigenvalues}
	Given $\mathcal{A} = (\alpha_{ij}) \in \mathbb{C}^{n \times n}$, $\mathcal{B} = (\beta_{ij}) \in \mathbb{C}^{n \times n}$, suppose that $\lambda(\mathcal{A}) =  \{\lambda_1(\mathcal{A}), \cdots, \lambda_n(\mathcal{A})\}$ and $\lambda(\mathcal{B}) = \{\mu_1(\mathcal{B}), \cdots, \mu_n(\mathcal{B})\}$ are the spectra of matrices $A$ and $B$ respectively, then
	\begin{equation}\notag
		{\rm sv}_{\mathcal{A}}(\mathcal{B}) \leq  \left( \| \mathcal{A} \| + \|\mathcal{B}\| \right)^{1-\frac{1}{n}} \| \mathcal{A} - \mathcal{B} \|^{\frac{1}{n}}.
	\end{equation}
\end{theorem}

Based on the definition of the Hausdorff distance and Theorem~\ref{perturbation of eigenvalues}, the following corollary can be directly derived.

\begin{corollary}\label{hausdorff_distance}
	Given $\mathcal{A} = (\alpha_{ij}) \in \mathbb{C}^{n \times n}$, $\mathcal{B} = (\beta_{ij}) \in \mathbb{C}^{n \times n}$, suppose that $\lambda(\mathcal{A}) =  \{\lambda_1(\mathcal{A}), \cdots, \lambda_n(\mathcal{A})\}$ and $\lambda(\mathcal{B}) = \{\mu_1(\mathcal{B}), \cdots, \mu_n(\mathcal{B})\}$ are the spectra of matrix $A$ and $B$ respectively, then the Hausdorff distance between the spectra of matrix $\mathcal{A}$ and $\mathcal{B}$ satisfy
	\begin{equation}\notag
		d_{\rm H}(\mathcal{A},\mathcal{B}) \leq  \left( \| \mathcal{A} \| + \|\mathcal{B}\| \right)^{1-\frac{1}{n}} \| \mathcal{A} - \mathcal{B} \|^{\frac{1}{n}}.
	\end{equation}
\end{corollary}

This corollary indicates that the Hausdorff distance between the spectra of matrix $\mathcal{A}$ and $\mathcal{B}$ can be controlled by the distance between matrix $\mathcal{A}$ and $\mathcal{B}$ in the sense of matrix norm.

\begin{lemma}\label{upper bound of Gamma and J}
    If the~\eqref{linear_system} is stable (or marginally stable), then the (extended) observability matrix $\Gamma_T$ (as defined in~\eqref{observability matrix}), the reversed (extended) controllability matrix $\mathcal K_T$ (as defined in~\eqref{controllability matrix}), the block Hankel matrix $\mathcal H_T$ (as defined in~\eqref{hankel matrix}), and the block Toeplitz matrix $\mathcal J_T$ (as defined in~\eqref{toeplitz}) satisfy
    \begin{equation}
        \begin{aligned}
        \| \Gamma_T \| \leq  \overline{c} \sqrt{Tmn}, \qquad
        \| \mathcal{K}_T \| \leq \overline{k} \sqrt{Tmn}, \\
        \| \mathcal{H}_T \|  \leq \overline{c}\overline{k}Tmn, \quad
        \| \mathcal{J}_T \| \leq Tm(1+\overline{c}\overline{k}n),
        \end{aligned}
    \end{equation}
    where $\overline{c} \triangleq \max_{i,j} |c_{ij}|$ and $\overline{k} \triangleq \max_{i,j} |k_{ij}|$ denote the maximum absolute values of the entries of matrices $C$ and $K$, respectively.
\end{lemma}

\begin{proof}
    We first bound $\| \Gamma_T \|$. Noting that $\| A^k \| \leq 1$ for all $k$ due to the stability (or marginal stability) of $A$, we have
    \begin{equation}\notag
\begin{aligned}
    \| \Gamma_T \|^2 &\leq \| \Gamma_T \|_{\mathrm{F}}^2 = \sum_{k = 0}^{T-1} \| C A^k \|_{\mathrm{F}}^2 
    \leq \sum_{k = 0}^{T-1} \| C \|_{\mathrm{F}}^2 \| A^k \|^2 \\
    &\leq \sum_{k = 0}^{T-1} \overline{c}^2 m n = \overline{c}^2 Tm n.
\end{aligned}
\end{equation}
This gives $\| \Gamma_T \| \leq \overline c \sqrt{Tmn}$. The bound for $\mathcal K_T$ is derived analogously:
$\| \mathcal K_T \| \leq \overline k \sqrt{Tmn}$. Since $\mathcal H_T = \Gamma_T \mathcal K_T$, it follows that
$\| \mathcal H_T \| \leq \| \Gamma_T \| \cdot \| \mathcal K_T \| \leq \overline{c}\overline{k}Tmn$.

For $\mathcal J_T$, we have
\begin{equation}\notag
\begin{aligned}
    \| \mathcal{J}_T \|^2 &\leq \| \mathcal{J}_T \|_{\mathrm{F}}^2 = T \| \mathbf{I}_m \|_{\rm F}^2 + \sum_{k = 0}^{T-2} (T-1-k) \| CA^k K\|_{\rm F}^2\\
    &\leq Tm + \sum_{k = 0}^{T-2} (T-1-k) \| C \|_{\mathrm{F}}^2 \| K \|_{\mathrm{F}}^2  \| A^k \|^2 \\
    &\leq Tm + \sum_{k = 0}^{T-2} (T-1-k) \overline{c}^2 \overline{k}^2 m^2n^2 \\
    &= Tm + \frac{T(T-1)}{2} \overline{c}^2 \overline{k}^2 m^2n^2 \leq (Tm + \overline{c}\overline{k}Tmn)^2,
\end{aligned}
\end{equation}
where the third inequality uses the assumption that \( A \) is stable or marginally stable, so \( \| A^k \| \leq 1 \) for all \( k \). Taking square roots completes the proof.
\end{proof}

\subsection{The proof of Theorem~\ref{estimate error of hankel matrix}}

Based on~\eqref{error_1}, the following three steps are essential for deriving an upper bound on the estimation error $\|\widehat{\mathcal{H}}_T - \mathcal{H}_T\|$:
\begin{enumerate}
	\item Providing condition for the invertibility of $Y_pY_p^\top$, i.e., the persistent of excitation for the batch past outputs $Y_p$.
	\item Giving an upper bound for the cross-term error $\mathcal{J}_TE_f Y_p^\top (Y_pY_p^\top)^{-1} $.
	\item  Giving an upper bound for the Kalman filter truncation bias term $\Gamma_T A_C^T \widehat{X} Y_p^\top (Y_pY_p^\top)^{-1}$.
\end{enumerate}

Next, we will present the results of the three steps sequentially.

\subsubsection{The invertibility of $Y_pY_p^\top$}

For each sample trajectory $i$, $i = 1,2,\cdots,N$, the past outputs can be expressed as
\begin{equation}\label{the past outputs}
	y_p^{(i)} = \Gamma_T \hat{x}_0^{(i)} + \mathcal{J}_T e_p^{(i)}.
\end{equation}
By stacking all trajectories, we obtain
\begin{equation}\label{Y_p}
	Y_p =  \Gamma_T \widehat{X} + \mathcal{J}_T E_p.
\end{equation}

\begin{lemma}
	Let $Y_p$ be as in \eqref{Y_p}. For a failure probability $\delta \in (0,1)$, let $N \geq N_1$, then with probability at least $1-\delta$, the following holds:
	\begin{equation}
		\mathcal{E}_1 = \left\{
		Y_pY_p^\top \succeq \frac{N\lambda_{\min}(\mathcal{R})}{2}  \mathbf{I}_{Tm} \succ \mathbf{0}
		\right\}
	\end{equation}
	where $N_1 = (6+4\sqrt{2})(\sqrt{Tm} + \sqrt{2\log(1/\delta)})^2$.
\end{lemma}

\begin{proof}
	Since $\hat{x}_0^{(i)}$ and the past noises $e_p^{(i)}$ are independent, the random variables $ y_p^{(i)}$ are i.i.d. Gaussian with mean zero and covariance matrix $ \Sigma_y = \Gamma_T P \Gamma_T^\top + \mathcal{J}_T (I_T \otimes \mathcal{S}) \mathcal{J}_T^\top$.
    
    On the other hand, from the system dynamics in the~\eqref{linear_system}, it follows that for each trajectory $i = 1, 2, \ldots, N$, the past outputs can also be expressed as
    \begin{equation}
	y_p^{(i)} = \Gamma_T x_0^{(i)} + \mathcal{M}_T w_p^{(i)} + v_p^{(i)},
    \end{equation}
    where the past process noise $w_p^{(i)}$ and the past measurement noise $v_p^{(i)}$ are given by
    \begin{equation}
	   w_p^{(i)} = \begin{bmatrix}
		w_0^{(i)} \\ \vdots \\ w_{T-1}^{(i)}
	\end{bmatrix}, \quad
	   v_p^{(i)} = \begin{bmatrix}
		v_0^{(i)} \\ \vdots \\ v_{T-1}^{(i)}
	\end{bmatrix}.
    \end{equation}
    and 
    $
    \mathcal{M}_T \triangleq \begin{bmatrix}
		\mathbf{I}_m & \\
		C & \mathbf{I}_m & \\
		\vdots & \vdots & \ddots & \\
		CA^{T-2} & CA^{T-3} & \cdots & \mathbf{I}_m
	\end{bmatrix}
    $.
    Since $x_0^{(i)}$, $w_p^{(i)}$, and $v_p^{(i)}$ are mutually independent, it follows that
    \begin{equation}
        \Sigma_y \succeq \mathrm{Cov}[v_p^{(i)}] = \mathbf{I}_T \otimes \mathcal{R}.
    \end{equation}
    Consequently, we obtain the lower bound
    \begin{equation}
        \lambda_{\min}(\Sigma_y) \geq \lambda_{\min}(\mathbf{I}_T \otimes \mathcal{R}) = \lambda_{\min}(\mathcal{R}) > 0.
    \end{equation}

    From the definition of $Y_p$, we have $Y_pY_p^\top  = \sum_{i=1}^{N} y_p^{(i)} (y_p^{(i)})^\top$. Applying Corollary~\ref{two-sided bound}, for $N \geq N_1$, with probability at least $1-\delta$, we obtain
    \begin{equation}
		\mathcal{E}_1 = \left\{
		Y_pY_p^\top \succeq \frac{N\lambda_{\min}(\mathcal{R})}{2}  \mathbf{I}_{Tm} \succ \mathbf{0}
		\right\}.
	\end{equation}
\end{proof}

\subsubsection{The cross-term error}

To bound the cross-term error, we need to derive an upper bound for $\|E_f Y_p^\top \| $. The product $E_f Y_p^\top$ can be written as 
\begin{multline}
    E_f Y_p^\top = \sum_{i=1}^N e_f^{(i)} (y_p^{(i)})^\top \\
    = \sum_{i=1}^N e_f^{(i)} (\hat{x}_0^{(i)})^\top \Gamma_T^\top +  \sum_{i=1}^N e_f^{(i)}(e_p^{(i)})^\top \mathcal{J}_T^\top,
\end{multline}
where the second equality follows from \eqref{the past outputs}. We then obtain the following bound: 
\begin{equation}\label{mid_cross}
	\| E_f Y_p^\top \| \leq \left\| \sum_{i=1}^N e_f^{(i)} (\hat{x}_0^{(i)})^\top \right\| \| \Gamma_T \| + \left\| \sum_{i=1}^N e_f^{(i)}  (e_p^{(i)})^\top  \right\| \| \mathcal{J}_T \|.
\end{equation}

\begin{lemma}
	Let $Y_p$ be as in \eqref{Y_p} and $E_f$ as defined in Section~\ref{sec:sim}. For a failure probability $\delta \in (0,1)$, if $N \geq N_2$, then with probability at least $1-2\delta$, the following holds:
	\begin{multline}
		\mathcal{E}_2 = \left\{
		\| E_f Y_p^\top \| \leq 4 \sqrt{2TmN\log (9/\delta)} \|\mathcal{S}\|^{1/2}\right.  \\
		\left. \cdot
        \left( \|P\|^{1/2} \overline{c} \sqrt{Tmn}
		+ \| \mathcal{S}\|^{1/2} Tm(1+\overline{c}\overline{k}n) \right)
		\right\}
	\end{multline}
	where $N_2 = Tm\log (9/\delta)$.
\end{lemma}

\begin{proof}
    Since $\hat{x}_0^{(i)}$ and the future noises $e_f^{(i)}$ are independent, and the future noises $e_f^{(i)}$ are independent of the past noises $e_p^{(i)}$, applying Lemma~\ref{xw} and Lemma~\ref{upper bound of Gamma and J} to~\eqref{mid_cross}, while noting that $Tm \geq nm \geq n$, completes the proof.
\end{proof}

\subsubsection{The truncation bias}

To bound the truncation bias, we need to derive an upper bound for $\|\widehat{X} Y_p^\top \| $. The product $\widehat{X} Y_p^\top$ can be written as
\begin{multline}
	\widehat{X} Y_p^\top = \sum_{i=1}^N \hat{x}_0^{(i)} (y_p^{(i)})^\top \\
    = \sum_{i=1}^N \hat{x}_0^{(i)} (\hat{x}_0^{(i)})^\top \Gamma_T^\top +  \sum_{i=1}^N \hat{x}_0^{(i)} (e_p^{(i)})^\top \mathcal{J}_T^\top,
\end{multline}
where the second equality follows from \eqref{the past outputs}. We then obtain the following bound:
\begin{equation}\label{mid_bias}
	\| \widehat{X} Y_p^\top \| \leq \left\| \sum_{i=1}^N \hat{x}_0^{(i)} (\hat{x}_0^{(i)})^\top \right\| \| \Gamma_T \| + \left\| \sum_{i=1}^N \hat{x}_0^{(i)} (e_p^{(i)})^\top  \right\| \| \mathcal{J}_T \|.
\end{equation}

\begin{lemma}
	Let $Y_p$ be as in \eqref{Y_p} and $E_f$ be as in \eqref{the state}. For a failure probability $\delta \in (0,1)$, if $N \geq N_3$, then with probability at least $1-2\delta$, the following holds:
	\begin{multline}
		\mathcal{E}_3 = \left\{
		\| \widehat{X} Y_p^\top \| \leq \sqrt{3}N \| P\| \overline{c} \sqrt{Tmn} \right.\\ \left.
        + 4\|P\|^{1/2}\| \mathcal{S}\|^{1/2} \sqrt{2TmN\log(9/\delta)} Tm(1+\overline{c}\overline{k}n)
		\right\},
	\end{multline}
	where $N_3 = \max\{(6+4\sqrt{2})(\sqrt{n} + \sqrt{2\log(1/\delta)})^2,N_2\}$.
\end{lemma}

\begin{proof}
    Since $\hat{x}_0^{(i)}$ and the past noises $e_p^{(i)}$ are independent, applying Corollary~\ref{two-sided bound}, Lemma~\ref{xw} and Lemma~\ref{upper bound of Gamma and J} to~\eqref{mid_bias} completes the proof.
\end{proof}

\subsubsection{Final proof}

\begin{proof}
    From~\eqref{error_1}, the estimation error $\|\widehat{\mathcal{H}}_T - \mathcal{H}_T\|$ can be bounded as follows:
\begin{multline}
	\|\widehat{\mathcal{H}}_T - \mathcal{H}_T\| \\
    \leq \|  (Y_pY_p^\top)^{-1}\| \left(
	\|\mathcal{J}_T\| \|E_f Y_p^\top\| +  \|\Gamma_T A_C^T \| \|\widehat{X} Y_p^\top\|
	\right).
\end{multline}
Let $N \geq \max\{N_1,N_2,N_3\} = (6+4\sqrt{2})(\sqrt{Tm} + \sqrt{2\log(1/\delta)})^2$. By applying the union bound to the estimates in the Lemmas with probability at least $1-5\delta$, the event $\mathcal{E}_1 \cap \mathcal{E}_2 \cap \mathcal{E}_3$ occurs. Therefore, we obtain the following bound:
\begin{equation}
	\|\widehat{\mathcal{H}}_T - \mathcal{H}_T\| \leq \frac{2}{ \lambda_{\min}(\mathcal{R})}  \left( \mathcal{C}_1 \mathcal{C}_2 + \sqrt{N}\|A_C^T\| \mathcal{C}_3 \right) \sqrt{\frac{T^5}{N}}.
\end{equation}
\end{proof}

\subsection{The proof of Lemma~\ref{the robustness of balanced realizations}}

\begin{proof}
	For the block Hankel matrix $\mathcal{H}_T$ and its best rank-$n$ approximation $\widehat{\mathcal{H}}_{T,n}$ of the estimated block Hankel matrix $\widehat{\mathcal{H}}_T$, one has
	\begin{equation}\label{mid_111}
		\| \mathcal{H}_T - \widehat{\mathcal{H}}_{T,n} \| \leq 	\| \mathcal{H}_T - \widehat{\mathcal{H}}_T \|  + \| \widehat{\mathcal{H}}_T - \widehat{\mathcal{H}}_{T,n} \| \leq 2\| \mathcal{H}_T - \widehat{\mathcal{H}}_T \|,
	\end{equation}
	where the second inequality holds because both $\mathcal{H}_T$ and $\widehat{\mathcal{H}}_{T,n}$ have rank $n$, and $\widehat{\mathcal{H}}_{T,n}$ is the best rank-$n$ approximation of $\widehat{\mathcal{H}}_T$, yielding $\| \widehat{\mathcal{H}}_T - \widehat{\mathcal{H}}_{T,n} \| \leq \| \mathcal{H}_T - \widehat{\mathcal{H}}_T \|$. Together with the perturbation condition~\eqref{perturbation condition}, this implies
	\begin{equation}
		\| \mathcal{H}_T - \widehat{\mathcal{H}}_{T,n} \|  \leq 2\| \mathcal{H}_T - \widehat{\mathcal{H}}_T \| \leq \frac{\sigma_{n}(\mathcal{H}_T)}{2}.
	\end{equation}
	Then, applying Theorem 5.14 of~\cite{tu2016lowranksolutionslinearmatrix}, there exists a unitary matrix $\mathcal{U} \in \mathbb{R}^{n \times n}$ such that
	\begin{equation}\label{mid_112}
		\| \widehat{\Gamma}_T - \overline{\Gamma}_T \mathcal{U} \|_{\rm F}^2 + \| \widehat{\mathcal{K}}_T - \mathcal{U}^\top \overline{\mathcal{K}}_T \|_{\rm F}^2 \leq \frac{2}{\sqrt{2}-1} \frac{\| \mathcal{H}_T - \widehat{\mathcal{H}}_{T,n} \|_{\rm F}^2}{\sigma_{n}(\mathcal{H}_T)}.
	\end{equation}
	Note that both $\mathcal{H}_T$ and $\widehat{\mathcal{H}}_{T,n}$ have rank $n$, thus $\mathcal{H}_T - \widehat{\mathcal{H}}_{T,n}$ has rank at most $2n$. Hence,
	\begin{equation}\label{mid_113}
		\| \mathcal{H}_T - \widehat{\mathcal{H}}_{T,n} \|_{\rm F} \leq \sqrt{2n} \| \mathcal{H}_T - \widehat{\mathcal{H}}_{T,n} \| \leq 2 \sqrt{2n} \| \mathcal{H}_T - \widehat{\mathcal{H}}_T \|,
	\end{equation}
	where the second inequality follows from~\eqref{mid_111}. 
    Since the spectral norm is bounded by the Frobenius norm, combining the above result with \eqref{mid_112}, we obtain
	\begin{equation}
		\| \widehat{\Gamma}_T - \overline{\Gamma}_T \mathcal{U} \|^2 + \| \widehat{\mathcal{K}}_T - \mathcal{U}^\top \overline{\mathcal{K}}_T \|^2 \leq \frac{16n}{\sqrt{2}-1} \frac{\| \mathcal{H}_T - \widehat{\mathcal{H}}_{T} \|^2}{\sigma_{n}(\mathcal{H}_T)}.
	\end{equation}
	Since $\frac{16}{\sqrt{2}-1} < 39 $, this further implies
	\begin{multline}
		\max\left\{
		\| \widehat{\Gamma}_T - \overline{\Gamma}_T \mathcal{U} \|, \| \widehat{\mathcal{K}}_T - \mathcal{U}^\top \overline{\mathcal{K}}_T \|
		\right\} \\
        \leq \sqrt{\frac{39n}{\sigma_{n}(\mathcal{H}_T)}} \| \mathcal{H}_T - \widehat{\mathcal{H}}_{T} \|.
	\end{multline}
	
	Since $\widehat{C} - \overline{C} \mathcal{U}$ is the first $m$ rows of $\widehat{\Gamma}_T - \overline{\Gamma}_T \mathcal{U} $, it immediately follows that
	\begin{equation}
		\| \widehat{C} - \overline{C} \mathcal{U} \| \leq \| \widehat{\Gamma}_T - \overline{\Gamma}_T \mathcal{U} \| \leq  \sqrt{\frac{39n}{\sigma_{n}(\mathcal{H}_T)}} \| \mathcal{H}_T - \widehat{\mathcal{H}}_{T} \|.
	\end{equation}
	Similarly, as $\widehat{K} -\mathcal{U}^\top \overline{K}$ is the last $m$ columns of $\widehat{\mathcal{K}}_T - \mathcal{U} ^\top \overline{\mathcal{K}}_T$, one has
	\begin{equation}
		\| \widehat{K} -\mathcal{U}^\top \overline{K} \| \leq \| \widehat{\mathcal{K}}_T - \mathcal{U} ^\top \overline{\mathcal{K}}_T \| \leq \sqrt{\frac{39n}{\sigma_{n}(\mathcal{H}_T)}} \| \mathcal{H}_T - \widehat{\mathcal{H}}_{T} \|.
	\end{equation}
	
	We now focus on $\widehat{A} - \mathcal{U}^\top \overline{A} \mathcal{U}$. Note that $\overline{A} = \overline{\Gamma}_{T,p}^\dagger \overline{\Gamma}_{T,f}$ and
	$\widehat{A} = \widehat{\Gamma}_{T,p}^\dagger \widehat{\Gamma}_{T,f}$, where $\overline \Gamma_{T,p}$ and $\overline \Gamma_{T,f}$ denote the matrices obtained from $\overline \Gamma_T$ by removing its last $m$ rows and first $m$ rows, respectively, and similarly for $\widehat \Gamma_{T,p}$ and $\widehat \Gamma_{T,f}$.
    Thus,
	\begin{multline}
		\widehat{A} - \mathcal{U}^\top \overline{A} \mathcal{U} = \widehat{\Gamma}_{T,p}^\dagger \widehat{\Gamma}_{T,f} - \mathcal{U}^\top \overline{\Gamma}_{T,p}^\dagger \overline{\Gamma}_{T,f} \mathcal{U} \\
		= \left( \widehat{\Gamma}_{T,p}^\dagger -  \mathcal{U}^\top \overline{\Gamma}_{T,p}^\dagger \right)  \overline{\Gamma}_{T,f} \mathcal{U}  +  \widehat{\Gamma}_{T,p}^\dagger \left( \widehat{\Gamma}_{T,f} - \overline{\Gamma}_{T,f} \mathcal{U}\right).
	\end{multline}
	It follows that
    \begin{multline}\label{mid_210}
        \| \widehat{A} - \mathcal{U}^\top \overline{A} \mathcal{U} \| \\
        \leq \| \widehat{\Gamma}_{T,p}^\dagger -  \mathcal{U}^\top \overline{\Gamma}_{T,p}^\dagger \| \cdot \|  \overline{\Gamma}_{T,f} \mathcal{U}\|  +\|  \widehat{\Gamma}_{T,p}^\dagger \| \cdot \| \widehat{\Gamma}_{T,f} - \overline{\Gamma}_{T,f} \mathcal{U}\| \\
			= \| \widehat{\Gamma}_{T,p}^\dagger -  \mathcal{U}^\top \overline{\Gamma}_{T,p}^\dagger \| \cdot \|  \overline{\Gamma}_{T,f}\|  +\|  \widehat{\Gamma}_{T,p}^\dagger \| \cdot \| \widehat{\Gamma}_{T,f} - \overline{\Gamma}_{T,f} \mathcal{U}\|,
    \end{multline}
    where the final equality uses the fact that $\mathcal U$ is a unitary matrix.

   First, since $\widehat{\Gamma}_{T,p}^\dagger $ and $ \mathcal{U}^\top \overline{\Gamma}_{T,p}^\dagger$ have the same rank $n$, applying Theorem 4.1 of~\cite{wedin1973perturbation} yields
   \begin{equation}\label{mid_211}
      \| \widehat{\Gamma}_{T,p}^\dagger -  \mathcal{U}^\top \overline{\Gamma}_{T,p}^\dagger \| \leq  \frac{\sqrt{2} \| \widehat{\Gamma}_{T,p} -  \mathcal{U}^\top \overline{\Gamma}_{T,p}\| }{\sigma_{n}(\widehat{\Gamma}_{T,p}) \sigma_{n}(\overline{\Gamma}_{T,p})}.
   \end{equation}
	Moreover, by the Cauchy interlacing theorem~\cite{hwang2004cauchy},
    \begin{equation}\label{mid_212_1}
            \sigma_{n}(\overline{\Gamma}_{T}) \leq \sigma_{n}(\overline{\Gamma}_{T,p}), \quad
            \sigma_{n}(\widehat{\Gamma}_{T}) \leq \sigma_{n}(\widehat{\Gamma}_{T,p}).
    \end{equation}
    On the other hand,
    \begin{equation}\label{mid_212_2}
        \sigma_{n}(\overline{\Gamma}_{T,p}) = \sqrt{\sigma_{n}(\mathcal{H}_T)}, \quad
        \sigma_{n}(\widehat{\Gamma}_{T,p}) = \sqrt{\sigma_{n}(\widehat{\mathcal{H}}_T)}.
    \end{equation}
	By combining \eqref{mid_211}, \eqref{mid_212_1}, \eqref{mid_212_2} and Lemma~\ref{sigma}, we can obtain that
    \begin{equation}\label{mid_213}
        \| \widehat{\Gamma}_{T,p}^\dagger -  \mathcal{U}^\top \overline{\Gamma}_{T,p}^\dagger \| 
        \leq \frac{2 \| \widehat{\Gamma}_{T,p} -  \mathcal{U}^\top \overline{\Gamma}_{T,p}\| }{\sigma_{n}(\mathcal{H}_T)}.
    \end{equation}
    
	Second, since $\overline \Gamma_{T,f}$ is formed from $\overline \Gamma_T$ by removing its first $m$ rows, it satisfies
	\begin{equation}\label{mid_214}
		\| \overline{\Gamma}_{T,f} \| \leq \| \overline{\Gamma}_T \| = \| \mathcal{H}_T \|^{1/2} \leq \sqrt{\overline{c}\overline{k}Tmn},
	\end{equation}
    where the last inequality follows from Lemma~\ref{upper bound of Gamma and J}.
	
	Third, it is not difficult to verify that
	\begin{multline}\label{mid_215}
		\max \left\{
		\| \widehat{\Gamma}_{T,p} -  \mathcal{U}^\top \overline{\Gamma}_{T,p}\| , \| \widehat{\Gamma}_{T,f} - \overline{\Gamma}_{T,f} \mathcal{U}\| 
		\right\} \leq 	
		\| \widehat{\Gamma}_T - \overline{\Gamma}_T \mathcal{U} \| \\
		\leq \sqrt{\frac{39n}{\sigma_{n}(\mathcal{H}_T)}} \| \mathcal{H}_T - \widehat{\mathcal{H}}_{T} \|.
	\end{multline}
    
	Finally, by substituting \eqref{mid_212_1}-\eqref{mid_215} into \eqref{mid_210}, we obtain
	\begin{equation}
		\| \widehat{A} - \mathcal{U}^\top \overline{A} \mathcal{U} \|  \leq \left(\sqrt{2} + 2\sqrt{
        \frac{\overline{c}\overline{k}Tmn}{\sigma_{n}(\mathcal{H}_T)}
        } \right) \frac{\sqrt{39n}}{\sigma_{n}(\mathcal{H}_T)}\|\mathcal{H}_T - \widehat{\mathcal{H}}_{T} \|.
	\end{equation}
\end{proof}

\subsection{The proof of Lemma~\ref{ill-conditioned}}

\begin{proof}
    Since the pair $(A, C)$ is observable and $T \ge n$, the matrix $\Gamma_T$ has full column rank $n$. Noting that $\mathcal H_T = \Gamma_T \mathcal K_T$, we have
    \begin{multline}
        \sigma_{n}(\mathcal{H}_T)  \leq \sigma_{n}(\Gamma_T) \| \mathcal{K}_T \| \leq 4\rho^{-\frac{\left\lfloor \frac{n-1}{2m} \right\rfloor}{\log (2mT)} }\|\Gamma_T\| \cdot\| \mathcal{K}_T \| \\
        \leq 4\overline{c}\overline{k} Tmn \rho^{-\frac{\left\lfloor \frac{n-1}{2m} \right\rfloor}{\log (2mT)} },
    \end{multline}
    where the second inequality follows from Lemma~\ref{krylov matrix}, and the final one is due to Lemma~\ref{upper bound of Gamma and J}.
\end{proof}

\subsection{The proof of Theorem~\ref{the system poles}}

\begin{proof}
    According to Theorem~\ref{end-to-end}, there exists a unitary matrix $\mathcal{U} \in \mathbb{R}^{n \times n}$ such that with probability at least $1-6\delta$, $\| \mathcal{U}^\top \widehat{A}\mathcal{U} - \overline{A} \| 
	\leq \Delta$.     
    On the other hand, according to the triangle inequality, we obtain
	\[
		\|\mathcal{U}^\top \widehat{A}\mathcal{U} \|  \leq \Delta + \| \overline{A}\|.
	\]
	Based on Corollary~\ref{hausdorff_distance}, it can be obtained that 
	\begin{multline}\notag
			d_{\rm H}(\widehat{A},\overline{A}) = d_{\rm H}(\mathcal{U}^\top \widehat{A}\mathcal{U},  \overline{A})  \\
			\leq  \left( \|\mathcal{U}^\top \widehat{A}\mathcal{U}\| + \|\overline{A}\| \right)^{1-\frac{1}{n}} \| \mathcal{U}^\top \widehat{A}\mathcal{U} - \overline{A}\|^{\frac{1}{n}} \\
            \leq (\Delta + 2\|\overline{A}\|)^{1-\frac{1}{n}} \Delta^{\frac{1}{n}},
	\end{multline}
	where the first equality holds because the unitary matrix transformation does not change the spectrum of the matrix. 
\end{proof}


\bibliographystyle{elsarticle-num}
\bibliography{ref}

\balance

\end{document}